\documentclass{llncs}
\usepackage{makeidx}
\usepackage{amssymb,amsmath}
\usepackage{amsfonts}
\usepackage{algorithmicx}

\DeclareMathOperator{\Char}{char}

\DeclareMathOperator{\gldim}{gl.dim }

\DeclareMathOperator{\tr}{tr}

\begin{document}


\newcommand{\Lra}{\Leftrightarrow}
\newcommand{\lra}{\longrightarrow}
\newcommand{\Ra}{\Rightarrow}
\newcommand{\F}{{\mathbb F}}
\newcommand{\N}{{\mathbb N}}
\newcommand{\K}{{\mathbb K}}
\newcommand{\C}{{\mathbb C}}
\newcommand{\R}{{\mathbb R}}
\newcommand{\Z}{{\mathbb Z}}
\newcommand{\Q}{{\mathbb Q}}
\newcommand{\E}{{\mathcal E}}
\newcommand{\nL}{{\mathcal L}}
\newcommand{\nNull}{{\mathrm {\mathbf 0}}}
\newcommand{\kg}{{\mathcal G}}
\newcommand{\ki}{{\mathcal I}}
\newcommand{\kv}{{\mathcal V}}
\newcommand{\ox}{{\overline{x}}}
\newcommand{\D}[1] {\partial_{#1}}
\renewcommand{\d}{{\partial}}
\newcommand{\mb}[1] {\mathbf{#1}}
\newcommand{\mT}{{\mathcal T}}
\newcommand{\kD}{{\mathcal D}}
\newcommand{\I}{{\mathcal I}}
\newcommand{\J}{{\mathcal J}}
\newcommand{\kb}{{\mathcal B}}
\newcommand{\A}{\ensuremath{\mathcal{A}}}
\newcommand{\B}{\ensuremath{\mathcal{B}}}
\newcommand{\g}{{\ensuremath{\mathfrak{g}}}}
\newcommand{\f}{{\ensuremath{\mathfrak{f}}}}
\newcommand{\gl}{{\mathfrak{ll}}}
\renewcommand{\t}{{\mathfrak{t}}}
\newcommand{\fI}{{\mathfrak{I}}}
\newcommand{\ra}{\rightarrow}
\newcommand{\GZS}{\operatorname{GZSupp}}
\renewcommand{\L}{{\mathcal{L}}}
\newcommand{\zsupp}{\ensuremath{\operatorname{Supp}}}
\newcommand{\n}{\ensuremath{\mathfrak{n}}}
\newcommand{\m}{\ensuremath{\mathfrak{m}}}
\newcommand{\ann}{\operatorname{Ann}}
\newcommand{\spec}{\operatorname{Spec}}
\newcommand{\Ann}{\operatorname{Ann}}
\newcommand{\rank}{\operatorname{rank}}
\newcommand{\syz}{\operatorname{Syz}}
\newcommand{\Syz}{\operatorname{Syz}}
\newcommand{\LeftSyz}{\operatorname{LeftSyz}}
\newcommand{\id}{\mathbb{I}\text{d}}
\newcommand{\algebraitem}[6] {\rule[-1mm]{0mm}{5mm}#1 & #2 & #3\\ \hline 
\multicolumn{3}{|l|}{\parbox[t]{120mm}{\setlength{\tabcolsep}{0mm}
\begin{tabular}{l@{\quad}rlllll}
 \rule{0mm}{5mm}Isomorphism: & $X$ & $\to\;$ & #4 , \; $Y$ & $\to\;$ & #5
\end{tabular}}}\\ 
\multicolumn{3}{|l|}{\parbox[t]{120mm}{#6}}\\ \hline}


\title{On Two-generated Non-commutative Algebras Subject to the 
  Affine Relation\thanks{
    This paper is contained in the Proceedings of CASC 2011, by Vladimir Gerdt, 
    Wolfram Koepf, Ernst W. Mayr, and Evgenii Vorozhtsov (eds.), Lecture Notes 
    in Computer Science, vol. 6885, ISBN 978-3-642-23567-2.
    The final publication is available at www.springerlink.com.}}
\titlerunning{Two-generated Algebras Subject to the Affine Relation}

\author{Viktor Levandovskyy\inst{1} and Christoph Koutschan\inst{2} and Oleksandr Motsak\inst{3}}
\authorrunning{V.~Levandovskyy, C.~Koutschan and O.~Motsak}
\institute{Lehrstuhl D f\"ur Mathematik, RWTH Aachen, Germany, \\ 
\email{viktor.levandovskyy@math.rwth-aachen.de} 
\and
RISC, Johannes Kepler University, Linz, Austria, \\
\email{Koutschan@risc.jku.at}
\and
TU Kaiserslautern, Germany, \\
\email{motsak@mathematik.uni-kl.de}}

\tocauthor{Viktor Levandovskyy, Christoph Koutschan, Oleksandr Motsak}

\maketitle

\begin{abstract}
We consider algebras over a field $\K$, generated by
two variables $x$ and $y$ subject to the single relation 
$yx = q xy + \alpha x + \beta y + \gamma$ for $q\in\K^*$ and
$\alpha, \beta, \gamma \in \K$. We prove, that among such algebras there are 
precisely five isomorphism classes. The representatives of these classes,
which are ubiquitous operator algebras, are called model algebras.
We derive explicit multiplication formulas for $y^m \cdot x^n$ in terms of standard
monomials $x^i y^j$ for many algebras of the considered type. 
Such formulas are used in e.~g. establishing formulas of binomial type and 
in an implementation of non-commutative multiplication in a computer algebra system. 
By using the formulas we also study centers and ring-theoretic properties of the 
non-commutative model algebras.
\end{abstract}



In this paper we study non-commutative algebras in two generators obeying single affine relation. Many operator algebras, coming from different areas of natural sciences, are built from algebras in two generators, see Sect. \ref{OpAlg} for examples. One of generators, say $x$, often corresponds to the operator of the multiplication with the function $x$. Another operator, say $y$, corresponds to a linear operator, acting on functions in the variable $x$. 

In the main Theorem we identify precisely five types of non-isomorphic algebras, which we call model algebras, among them. Despite the fact that many such algebras have been studied in the literature (see e.~g. \cite{Dix,MR,Mgfun,BGV}, many aspects and properties are too scattered in the existing literature. Another point of this note is to search systematically for closed form of multiplication formulas on monomials. Such closed forms are needed, among other, in computer algebra, where many sophisticated algorithms heavily rely on basic multiplication among monomials. It is not enough to have such formulas just for model algebras, since isomorphisms do not preserve monomials but turn them into polynomials. It turned out, that there are still several cases, where we were not able to derive closed formulas in terms of standard monomials. With our approach, however, one is still able to derive formulas of certain type for them.

\section{Preliminaries}

Let $\K$ be a field. Moreover, let $A$ be an associative $\K$-algebra and $q\in\K^*$. We use the following notations: 
$[a,b]_q := ab - q\cdot ba$ is a \textit{$q$-commutator} of $a,b \in A$. 
The \textit{commutator} or the \textit{Lie bracket} 
is $[a,b] := [a,b]_1 = ab - ba$. 
We also write $[n] = [n]_q = \tfrac{q^n-1}{q-1}$ for the $q$-number, 
$(a;q)_n := \prod_{k=0}^{n-1} (1-aq^k)$ for the $q$-Pochhammer symbol,
$[n]^{\underline{k}} = \tfrac{(q^{n-k+1};q)_k}{(1-q)^k}$ for the $q$-falling factorial
and $\left[ \begin{array}{c} n \\ k \end{array} \right] = 
\left[ \begin{array}{c} n \\ k \end{array} \right]_q = \tfrac{[n]!}{[n-k]![k]!}$ for the $q$-binomial coefficient. Note, that $\left[\begin{array}{c}n\\ k\end{array}\right]=0$ for $k>n$.

\begin{lemma}
 $\forall a,b,c \in A$ and $\lambda,\mu\in\K$ the following identities hold.
\begin{itemize}
\item $[a,b]_q = -q(ba - \frac{1}{q} ab) = -q[b,a]_{q^{-1}}$, \;  $[a,a]_q=(1-q)a^2$ 
\item $[a+\lambda,b]_q = [a,b]_q + \lambda (1-q)b$, \; $[a,b+\mu]_q = [a,b]_q + \mu (1-q)a$
\item $[ab,c]_q = a[b,c]_q + q \cdot [a,c] b = a [b,c] + [a,c]_q b$
\end{itemize}
In particular, $[a,b]  = -[b,a]$ and $[ab,c]=a[b,c]+[a,c]b$.
\end{lemma}



\noindent
We study two-generated non-commutative $\K$-algebras with affine relations
\[
A(q,\alpha, \beta, \gamma) := \K \langle x,y \mid yx = q\cdot xy + \alpha x + \beta y + \gamma \rangle
\]
for $q \in \K^{*}$ and $\alpha, \beta, \gamma \in \K$. The scalar $q$ plays an important role and we distinguish two major cases. If $q=1$, an algebra is of \textit{Lie type}, that is it is isomorphic 
to a factor-algebra of the universal enveloping algebra of a finite-dimensional Lie algebra. 
If $q\not=1$, in the research of \textit{quantum algebras} one distinguishes
two situations (which lead to different behaviour of algebras): either $q$ is transcendental over some subfield $k\subset\K$ or $q$ is a root of unity in $\K$. Without assumptions on $q$ we will write $\K(q)$ in general (thus encompassing the case $q\in\K^*$ as well), while in the case $q=1$ just $\K$ will be used.
For $(a,b) \in \N_0^2$ we call an element $x^a y^b$ a (standard) \textbf{monomial}. 
If an algebra $A$ possesses a $\K$-basis, consisting of monomials, the latter basis is often called a \textbf{Poincar\'e-Birkhoff-Witt} basis. 
The following Lemma is well-known.
\begin{lemma}
\label{PBWA}
$A(q,\alpha, \beta, \gamma)$ has $\{x^a y^b \mid (a,b) \in \N_0^2\}$ as a $\K(q)$-basis. 
\end{lemma}

 Indeed, this Lemma is a consequence of the more general statement, which
 can be easily proved by using Gr\"obner bases in the free associative algebra 
 $\K\langle x,y \rangle$. The latter algebra has a $\K$-basis, consisting of \textit{words},
 that is of elements from the free monoid $F = \langle x,y \rangle$. 
 The empty word from $F$ is written as $1$ in $\K\langle x,y \rangle$. 
 The free monoid $F$ can be totally well-ordered with an ordering, 
 which is compatible with the bilateral multiplication on $F$. We call an ordering,
 having such properties, a \textit{monomial} ordering on $\K\langle x,y \rangle$.
 See e.~g. \cite{Mora,Gr,Ufn98} for the Gr\"obner bases theory for $\K\langle x,y \rangle$.

 \begin{lemma}
 Let $\K\langle x,y \rangle$ be the free associative algebra and $\prec$ be a monomial
 ordering. Consider a polynomial $p = c\cdot yx + q(x,y)$, $c\in\K\setminus\{0\}$ such that every word of $q(x,y)$ is smaller than $yx$ with respect to $\prec$. Then $\{ p \}$ is a two-sided Gr\"obner basis of the two-sided ideal $\langle p \rangle$.
 \end{lemma}

 \begin{proof}
 Since there are no words $u,v\in F$ of degree less than 2, such that $u\cdot yx=yx \cdot v$ holds, 
 there are no generalized $s$-polynomials in the free algebra~\cite{Gr}. Hence the set $\{ p \}$ is already a two-sided Gr\"obner basis.
 \end{proof}

 Writing $A(q,\alpha, \beta, \gamma) \cong \K\langle x,y \rangle / \langle -yx + q\cdot xy + \alpha x + \beta y + \gamma  \rangle$, we see that the previous Lemma proves Lemma \ref{PBWA}. Consider a monomial ordering, which satisfies (a) $w \prec u$ implies $\deg w \leq \deg u$ and (b) $xy \prec yx$, for instance, degree right lexicographical ordering. Then by the previous Lemma $\{ -yx + q\cdot xy + \alpha x + \beta y + \gamma \}$ is a two-sided Gr\"obner basis, hence, the basis of the factor algebra 
 $A(q,\alpha, \beta, \gamma)$ is spanned by all words
 which do not contain $yx$ as a subword. And such words are precisely the standard monomials.

A product of two monomials is, in general, not a monomial, but a
polynomial, that is a sum of monomials.

\section{Main Theorem and Applications}

\begin{theorem}
\label{MainThm}
$A(q,\alpha, \beta, \gamma)$ is isomorphic to one of the five \textbf{model algebras}:
\begin{enumerate}
\item the commutative algebra $\K[x,y]$,
\item the first Weyl algebra $A_1 = \K\langle x,d \mid dx = xd +1 \rangle$ (the algebra of linear differential operators with coefficients from $\K[x]$),
\item the shift algebra $S_1 = \K\langle x,s \mid sx = xs + s \rangle$ (the universal enveloping algebra of the non-abelian solvable two-dimensional Lie algebra; the algebra of linear shift operators with coefficients from $\K[x]$),
\item[]
\item the $q$-commutative algebra $\K_q[x,y]:=\K(q) \langle x,y \mid yx = q\cdot xy \rangle$ 
(Manin's quantum plane; the algebra of linear $q$-shift operators with coeff's from $\K(q)[x]$) 
\item the first $q$-Weyl algebra $A_1^{(q)} = \K(q) \langle x,\partial \mid \partial x = q \cdot x \partial +1 \rangle$ (the algebra of linear $q$-differential operators with coefficients from $\K(q)[x]$).
\end{enumerate}
Moreover, the model algebras are pairwise non-isomorphic (see Prop. \ref{Modelz}).
\end{theorem}

In Tables 1 and 2
we write isomorphisms to model algebras 
and write formulas for the multiplication in every concrete
class of algebras. In some cases we also write down the recurrence
formulas for the coefficients in the expansion of $y^m x^n$ in terms
of standard monomials $x^a y^b$. For some algebras we put simpler
formulas for $y^m x$ and $y x^n$ as well as a part of our proof.

By writing \textbf{not known yet} in the table we mean, that up to now, no
explicit formula in terms of of standard monomials is known to us. 
However, by applying an isomorphism (for instance, the one we give explicitly in the table)
to the explicit multiplication formula of the corresponding model algebra (Algebra Class in the table),
we obtain a non-expanded formula for any algebra in the table.

\begin{table}
\label{tabLie}
\caption{Multiplication Formulas for Algebras of Lie Type}
\begin{center}
\begin{tabular}{|c|c|c|}
\hline
Algebra Type  &  Relation & Algebra Class\\
\hline\hline
\algebraitem{$(1,0,0,0)$}
{$yx=xy$}
{$YX=XY$}
{$x$}{$y$}
{$\displaystyle y^mx^n=x^ny^m$}

\algebraitem{$(1,\alpha,0,0)$}
{$yx=xy+\alpha x$}
{$YX=XY+Y$}
{$-\alpha^{-1}y$}{$x$}
{$\displaystyle y^mx^n=x^n(y+n \alpha)^m=\sum_{k=0}^m \binom{m}{k} (n\alpha)^{m-k} x^{n} y^{k}$,\\
Coeff. recurrence: $\displaystyle C_k = \frac{(k+1)n\alpha}{m-k} C_{k+1}$}

\algebraitem{$(1,0,\beta,0)$}
{$yx=xy+\beta y$}
{$YX=XY+Y$}
{$\beta^{-1}x$}{$y$}
{$\displaystyle y^mx^n=(x+m\beta)^ny^m=\sum_{k=0}^n \binom{n}{k}(m\beta)^{n-k}x^ky^m$,\\
Coeff. recurrence: $\displaystyle C_k = \frac{(k+1)m\beta}{n-k} C_{k+1}$}

\algebraitem{$(1,\alpha,\beta,0)$}
{$yx=xy+\alpha x+\beta y$}
{$YX=XY+Y$}
{$-\alpha^{-1}y$}{$\alpha x+\beta y$}
{$\displaystyle yx^n=\frac{1}{\beta}\big((x+\beta)^n(\alpha x+\beta y)-\alpha x^{n+1}\big)$,\;
$\displaystyle y^mx=\frac{1}{\alpha}\big((\alpha x+\beta y)(y+\alpha)^m-\beta y^{m+1}\big)$,\\
$\displaystyle y^mx^n=$ \; not known yet}


\algebraitem{$(1,0,0,\gamma)$}
{$yx=xy+\gamma$}
{$YX=XY+1$}
{$x$}{$\gamma^{-1}y$}
{$\displaystyle yx^n=x^{n-1}(xy+n\gamma)$,\;
$\displaystyle y^mx=(xy+m\gamma)y^{m-1}$,\\
$\displaystyle y^mx^n=\sum_{k=0}^n \binom{m}{k}n^{\underline{k}}\gamma^kx^{n-k}y^{m-k}
=\sum_{k=0}^{\min\{m,n\}}\frac{m!n!\gamma^kx^{n-k}y^{m-k}}{k!(m-k)!(n-k)!}$,\\
Coeff. recurrence: $\displaystyle C_k=\frac{(m-k+1)(n-k+1)\gamma}{k}C_{k-1}$}

\algebraitem{$(1,\alpha,0,\gamma)$}
{$yx=xy+\alpha x+\gamma$}
{$YX=XY+Y$}
{$-\alpha^{-1}y$}{$\alpha x+\gamma$}
{$\displaystyle yx^n=x^ny+nx^{n-1}(\alpha x+\gamma)$,\;
$\displaystyle y^mx=\frac{1}{\alpha}\big((\alpha x+\gamma)(y+\alpha)^m-\gamma y^m\big)$,\\
$\displaystyle y^mx^n=\frac{1}{\alpha^n}\sum_{i=0}^n \binom{n}{i}(-\gamma)^{n-i}(\alpha x+\gamma)^i(y+i\alpha)^m$}

\algebraitem{$(1,0,\beta,\gamma)$}
{$yx=xy+\beta y+\gamma$}
{$YX=XY+Y$}
{$\beta^{-1}x$}{$\beta y+\gamma$}
{$\displaystyle y^mx=xy^m+my^{m-1}(\beta y+\gamma)$,\;
$\displaystyle yx^n=\frac{1}{\beta}\big((x+\beta)^n(\beta y+\gamma)-\gamma x^n\big)$,\\
$\displaystyle y^mx^n=\frac{1}{\beta^m}\sum_{i=0}^m \binom{m}{i}(-\gamma)^{m-i}(x+i\beta)^n(\beta y+\gamma)^i$}

\algebraitem{$(1,\alpha,\beta,\gamma)$}
{$yx=xy+\alpha x+\beta y+\gamma$}
{$YX=XY+Y$}
{$-\alpha^{-1}y$}{$\alpha x+\beta y+\gamma$}
{$\displaystyle y^mx^n=$ \; not known yet}

\end{tabular}
\end{center}
\end{table}

\begin{table}
\label{tabQ}
\caption{Multiplication Formulas for Quantum Algebras}
\begin{center}
\begin{tabular}{|c|c|c|}
\hline
Algebra Type  &  Commutation & Algebra Class\\
\hline\hline

\algebraitem{$(q,0,0,0)$}{$yx=qxy$}{$YX=qXY$}{$x$}{$y$}
{$\displaystyle y^mx^n=q^{mn}x^ny^m$}

\algebraitem{$(q,\alpha,0,0)$}{$yx=qxy+\alpha x$}{$YX=qXY$}{$x$}{$y-\alpha(1-q)^{-1}$}
{$\displaystyle y^mx^n=x^n(q^ny+[n]\alpha)^m$}

\algebraitem{$(q,0,\beta,0)$}{$yx=qxy+\beta y$}{$YX=qXY$}{$x-\beta(1-q)^{-1}$}{$y$}
{$\displaystyle y^mx^n=(q^mx+[m]\beta)^ny^m$}

\algebraitem{$(q,\alpha,\beta,0)$}{$yx=qxy+\alpha x+\beta y$}{$YX=qXY$}{$x-\beta(1-q)^{-1}$}{$y-\alpha(1-q)^{-1}$}
{$\displaystyle y^mx=x(qy+\alpha)^m+\beta\sum_{k=1}^my^k\alpha^{m-k}\sum_{i=0}^{k-1}\binom{m-k+i}{i}q^i$,\\
$\displaystyle y^mx^n=$ \; not known yet}

\algebraitem{$(q,0,0,\gamma)$}{$yx=qxy+\gamma$}{$YX=qXY+1$}{$x$}{$\gamma^{-1}y$}
{$\displaystyle y^mx^n=\sum_{k=0}^n\left[\begin{array}{c}m\\ k\end{array}\right][n]^{\underline{k}}
q^{(n-k)(m-k)}\gamma^kx^{n-k}y^{m-k}$.}

\algebraitem{$(q,\alpha,0,\gamma)$}{$yx=qxy+\alpha x+\gamma$}{$YX=qXY+1$}
{$\gamma^{-1}x$}{$y-\alpha(1-q)^{-1}$}
{$\displaystyle y^mx^n=\sum_{k=0}^n\sum_{j=0}^{m-k}\left[\begin{array}{c}n\\ k\end{array}\right]
\gamma^k\left(\frac{\alpha}{1-q}\right)^{m-j-k}c_{j,k,m,n}x^{n-k}y^j$,\\
where $\displaystyle c_{j,k,m,n}=
\sum_{i=0}^{m-j-k}(-1)^i\binom{m}{i+j+k}\binom{i+j}{j}[i+j+k]^{\underline{k}}q^{(i+j)(n-k)}$}

\algebraitem{$(q,0,\beta,\gamma)$}{$yx=qxy+\beta y+\gamma$}{$YX=qXY+1$}
{$x-\beta(1-q)^{-1}$}{$\gamma^{-1}y$}
{$\displaystyle y^mx^n=\sum_{k=0}^n\sum_{j=0}^{n-k}\left[\begin{array}{c}m\\ k\end{array}\right]
\gamma^k\left(\frac{\beta}{1-q}\right)^{n-j-k}c_{j,k,m,n}x^jy^{m-k}$,\\
where $\displaystyle c_{j,k,m,n}=
\sum_{i=0}^{n-j-k}(-1)^i\binom{n}{i+j+k}\binom{i+j}{j}[i+j+k]^{\underline{k}}q^{(i+j)(m-k)}$}

\algebraitem{$(q,\alpha,\beta,\gamma)$}{$yx=qxy+\alpha x+\beta y+\gamma$}{$YX=qXY+1$}
{$x-\beta (1-q)^{-1}$}{$((1-q) y - \alpha)(\gamma (1-q) + \alpha \beta)^{-1}$}
{$\displaystyle y^mx^n=$ \; not known yet}

\end{tabular}
\end{center}
\end{table}

\begin{proof}
While for some of the above cases the explicit formulas for $y^mx^n$ are rather simple
(and therefore easily found), others are quite complicated and required some work.
A good strategy for finding a general formula for $y^mx^n$ is to study the special
cases $yx^n$ and $y^mx$ first. Once this is done, further multiplications by $y$
(and $x$, respectively) lead to the general formula. However, for the most general
commutation rules (e.g., $yx=xy+\alpha x+\beta y+\gamma$), this strategy fails. 

All the formulas for $y^mx$, $yx^n$, and $y^mx^n$ have in common that they are 
easily proved by induction. As an example, consider the algebra $(1,0,\beta,\gamma)$.
We have stated above that $y^mx=xy^m+my^{m-1}(\beta y+\gamma)$. For $m=1$ this 
reduces just to the given commutation relation $yx=xy+\beta y+\gamma$. Now
consider $y^{m+1}x=y\cdot(y^mx)$ which by induction hypothesis is 
$yxy^m+my^m(\beta y+\gamma)=(xy+\beta y+\gamma)y^m+my^m(\beta y+\gamma)$
which after collecting powers gives the desired formula $xy^{m+1}+(m+1)y^m(\beta y+\gamma)$.
Similarly, the general formula
\[
  y^mx^n=\frac{1}{\beta^m}\sum_{i=0}^m\binom{m}{i}(-\gamma)^{m-i}(x+i\beta)^n(\beta y+\gamma)^i
\]
can be shown (now we use induction on~$n$). 
A straightforward calculation shows that this formula for $n=1$ reduces
to the one given above for $y^m x$.
Thus it has to be investigated what happens
after multiplying another~$x$ from the right:
\begin{eqnarray*}
(\beta y+\gamma)^ix & = & \sum_{j=0}^i\binom{i}{j}(\beta y)^j\gamma^{i-j}x
 =  \sum_{j=0}^i\binom{i}{j}\beta^j\left(xy^j+jy^{j-1}(\beta y+\gamma)\right)\gamma^{i-j}\\
& = & x\sum_{j=0}^i\binom{i}{j}\beta^jy^j\gamma^{i-j}+\beta\left(\sum_{j=0}^i\binom{i-1}{j-1}i(\beta y)^{j-1}\gamma^{i-j}\right)(\beta y+\gamma)\\
& = & x(\beta y+\gamma)^i + \beta i(\beta y+\gamma)^{i-1}(\beta y+\gamma)\\
& = & (x+\beta i)(\beta y+\gamma)^i
\end{eqnarray*}
\end{proof}

We have additionally checked the validity of the formulas above with our respective implementations in
computer algebra systems \textsc{Singular:Plural} \cite{LS} and \textsc{Mathematica} \cite{HolFun}.

\subsection{Operator Algebras and Model Algebras}
\label{OpAlg}

Fix a constant $c \in\K^*$. Then the $c$-\textit{shift operator} acts as $s_c(f(x)) = f(x-c)$.
The corresponding $c$-\textit{shift algebra} is $\K \langle x,s_c \ \mid \ s_c \cdot x = x \cdot s_c - c s_c \rangle$.
For $c=1$ one recovers discrete shift operator. If $c <0$ (resp. $c >0$), $s_c$ is called an \textit{advance operator} (resp. a \textit{time-delay operator}) in both discrete and continuous settings. The corresponding algebras are of the type $(1,0,\beta=-c,0)$ and thus they are
isomorphic to $\K\langle X,Y \mid YX = XY + Y \rangle$, the model shift algebra.

Let $c = (c_1,c_2)$ for $c_i\in\K^*$. The $c$-\textit{difference operator} acts as $\Delta_c(f(x)) = \tfrac{f(x+c_1)-f(x)}{c_2}$.
The corresponding $c$-\textit{difference algebra} is
\[
\K \langle x,\Delta_c \ \mid \ \Delta_c \cdot x = x \cdot \Delta_c + c_1 \Delta + \tfrac{c_1}{c_2} \rangle.
\]
For $c=(1,1)$ one recovers discrete difference operator; for $c=(\triangle x, \triangle x)$ the first-order divided difference operator. The corresponding algebras are of the type $(1,0,\beta=c_1,\gamma=c_1 c_2^{-1})$ and hence they are isomorphic to $\K\langle X,Y \mid YX = XY + Y \rangle$, the model shift algebra.


Following Chyzak and Salvy \cite{Mgfun}, the $q$-\textit{dilation}
and $q$-\textit{shift} operators give rise to the same operator algebra, the $q$-commutative model algebra $\K_q[x,y]$. Both continuous and discrete $q$-\textit{difference} operators \cite{Mgfun} give rise to the algebra 
$\K(q)\langle x,y \mid yx=qxy + (q-1)x\rangle$ of the type $(q,\alpha=q-1,0,0)$. Hence it is
isomorphic to the $q$-commutative model algebra $\K_q[x,y]$. 

Let $c = (c_1,c_2)$ for $c_i\in\K(q)^*$ with $q^{c_i}\neq 0$. The $c$-$q$-\textit{differential operator} acts as $\Delta^{(q)}_c(f(x)) = \tfrac{f(q^{c_1}x)-f(x)}{(q^{c_2}-1)x}$.
The corresponding $c$-$q$-\textit{differential algebra} is
\[
\K(q) \langle x,\Delta^{(q)}_c \ \mid \ \Delta_c \cdot x = q^{c_1} x \cdot \Delta_c + (q^{c_1}-1) \cdot (q^{c_2}-1)^{-1} \rangle.
\]
For $c=(1,1)$ one recovers the $q$-\textit{differential operator} $D_q(f(x)) = \tfrac{f(qx)-f(x)}{qx -x}$. Otherwise, we use Table 2 and by sending $x\to X, \Delta^{(q)}_c \to Y:=(q^{c_2}-1)(q^{c_1}-1)^{-1} \Delta^{(q)}_c$ we obtain the isomorphic algebra $\K(q) \langle X,Y \mid YX = q^{c_1} XY + 1 \rangle$. 
Let $\tilde{q}=q^{c_1}$, then the subalgebra $\K(\tilde{q}) \langle X,Y \mid YX = \tilde{q} XY + 1 \rangle$ of the previous algebra is the first $\tilde{q}$-Weyl model algebra.

Consider the differentiation $y=\tfrac{d}{dt}$ and the operator $x(f(t)) := e^{\lambda t} \cdot f(t)$ for $\lambda\in\K^*$. Then the algebra, generated by $x,y$ has the relation 
$yx = xy + \lambda x$ and it is isomorphic to the model shift algebra.

Of course, there are operators obeying relations, which are not affine. Consider the \textit{integration} operator $I(f(x)):= \int_0^x f(t) dt$. Its relation with $x$ reads as $Ix = xI - I^2$.
Similarly, let $x = t^{-1}$ and $y=\tfrac{d}{dt}$. Then the relation is 
$yx = xy - x^2$. Both algebras can be realized as $G$-algebras. It is interesting to study 
 model algebras for non-affine relations.

\begin{remark}
Note, that isomorphy of $q$-shift and $q$-commutative algebras does not
have an analogue in the classical situation, since for $q=1$ the model shift algebra is
not isomorphic to the model commutative algebra. Thus the following question arises: is there a quantum algebra (clearly, with non-affine relation), which becomes shift model algebra in the limit $q\to 1$?
\end{remark}

\subsection{Binomial Theorems}

Notation: in a noncommutative algebra $A$, for two elements $a,b \in A\setminus\K$, we define $[a+b]^n := \sum_{i=0}^n \binom{n}{i} a^i b^{n-i}$. 
Respectively, we define 
$[a+b]_q^n := \displaystyle \sum_{i=0}^n \left[ \begin{array}{c} n \\ i \end{array} \right]_q a^i b^{n-i}$. 
Then, if $x,y$ commute, one expresses
the binomial theorem as $(x+y)^n = [x+y]^n$. Respectively, if $yx = qxy$, we obtain
$(x+y)^n = [x+y]_q^n$.

Using the formulas obtained above, we can provide formulas of binomial type,
which are important in applications. Among the variety of possible presentations
in such formulas we aim at those, which express $(x+y)^n$ in terms of standard monomials $x^i y^j$.

In the free associative algebra $\K\langle a,b \rangle$, we can write 
$(a+b)^n = \sum_{w\in \langle a,b \rangle_n} w$, that is $w$ run through all words
of length $n$ in the free monoid $\langle a,b \rangle$. One defines a \textit{misordering
index} \cite{BGV} of $w$ to be the number of operations, each of them exchanges two neighbour 
non-equal letters, needed to move all $a$'s to the left (thus finishing when a standard monomial has been achieved), 
starting from the last letter in $w$. For example, the misordering index of a standard monomial
is $0$, while the misordering index of $bbbab$ is $3$, since the sequence of exchange operations
is $bbbab,bbabb,babbb,abbbb$. We say also that $bbbab$ \textit{converges} to $abbbb$ here. 
It is known, that in any algebra $A(q,\alpha, \beta, \gamma)$ the leading monomial of a polynomial 
$y^m\cdot x^n$ is $x^n y^m$. Hence, the coefficients of a standard monomial $x^a y^b$ of degree $a+b$
will appear from the multiplication, applied on every word, which converges to $x^a y^b$. And closed 
formulas for multiplication allow to perform this task symbolically.

\begin{lemma}
Let $A = A_1$ be the first Weyl algebra, where $\d x = x \d + 1$ holds. 
Then the following binomial theorem takes place:
$$(x+d)^n-[x+d]^n=\sum_{k=0}^{n-2}\sum_{j=0}^{n-k-2}\binom{n}{j}\binom{n-j}{k}g(n-j-k)x^kd^j$$
where $g(n):=(n-1)!!$, if $n$ is even and $0$ otherwise.
Alternatively we can write\\
$\displaystyle (x+d)^n-[x+d]^n=\sum_{0\leq k\leq n-2}\sum_{\genfrac{}{}{0pt}{2}{0\leq j\leq n-k-2}{n-j-k\mathrm{\ even}}}
\binom{n}{j}\binom{n-j}{k}(n-j-k-1)!!x^kd^j$\\
$\displaystyle =\sum_{0\leq k\leq n-2}\sum_{\genfrac{}{}{0pt}{2}{0\leq j\leq n-k-2}{n-j-k\mathrm{\ even}}}
\frac{n!}{j!k!(\frac{n-j-k}{2})!}\left(\frac{1}{2}\right)^{\frac{n-j-k}{2}}x^k d^j$.
\end{lemma}

\begin{lemma}
For the shift algebra $S_1$, where $xs=sx+s$ holds, we obtain the following binomial theorem:
$$(x+s)^n=[x+s]^n+\sum_{k=0}^{n-1}\sum_{j=0}^{n-k-1}\binom{n}{k}S(n-k,j)x^ks^j$$
where $S(n,k)$ denote the Stirling numbers of the second kind.
\end{lemma}

We omit the technical proofs for both Lemmas. They can be done by induction, using the multiplication
formulas.


\section{Application in Computer Algebra Implementation}

As described in \cite{LS}, a general multiplication in a non-commutative $G$-algebra
boils down to the multiplication of $y^m \cdot x^n$ for a couple of variables $x,y$
such that $x^a y^b$ is a standard word. 
In general, the polynomial $y^m \cdot x^n$ involves other variables as well,
but the case, when $x,y$ generate a subalgebra of the type $A(q,\alpha, \beta, \gamma)$,
appears often enough. Suppose from now on we are in such situation.

In \cite{LS} it has been proposed to address each pair of non-commuting and non-$q$-commuting
variables separately. To each such pair a matrix $M$ is assigned, such that $M_{ij}=y^i \cdot x^j$
is a polynomial, written in terms of standard monomials. There is a general multiplication algorithm,
which uses matrix entries of lower degree in order to compute the higher degrees on demand.

There are several different strategies on the usage of the formulas for enhancing the polynomial
multiplication. Of course, this problem barely has an analogue in the commutative case.
Initialization of non-commutative relation between $y$ and $x$ saves the relation $yx=q\cdot xy + \alpha x + \beta y + \gamma$ as a part of data structure on the algebra, where the computations take place.

\textbf{1. Faster computation, considerable memory usage}: As proposed in \cite{LS}, 
the results of all required multiplications $y^i \cdot x^j$ and the intermediate
multiplications in lower degree will be saved. Due to the same principles applied for
the search of previously computed elements of lower degree, the multiplication matrix
will be filled with many elements. On the other hand, the intermediate elements 
will be reused intensively and this leads to fast arithmetics in the algebra.

\textbf{2. Saving memory, slower computation}: All required multiplications $y^i \cdot x^j$ will
be done according to the formulas, the results will not be saved for the future use. 
Thus this way uses the least amount of memory, but can take much longer, especially if
many multiplications are requested repeatedly.

\textbf{3. Mixing 1 and 2 and using formulas}: Computing, by utilizing formulas, 
the requested elements and storing them into the multiplication matrix eliminates 
the need to compute and store intermediate elements from the approach 1. 
Storing the demanded elements increases the chances for the future reuse of matrix entries.
Still, there are more possibilities to develop strategies by mixing both approaches and
working with multiplication matrices dynamically, like keeping (e.~g. by periodic cleaning) 
higher degree part of the matrix sparse while being as dense as possible in the lower
degree part. But the question, how to determine the value, which distinguishes high degree
from low degree, is open. At last, but not at least, we have experimented with 
counting the requests to each matrix entry, thus having a metric for the usability of every single entry. This is useful while following the strategy, which uses periodic cleaning of multiplication matrices.

\textbf{Experiments}. 
Let us do experiments with the most general case: $A = A(q, \alpha, \beta, \gamma)$, 
where $q, \alpha, \beta, \gamma$ are transcendental over the base field.
As described above, in general the product $y^a \cdot x^b$ is computed (inductively) either as 
$y \cdot (y^{a-1} \cdot x^b)$ or as $(y^a \cdot x^{b-1}) \cdot x$.
Let us consider the products $y^i \cdot x$ and $y \cdot x^i$. The determination of
the computational method for these products can be made during run-time by
analyzing given $q, \alpha, \beta, \gamma$.
Both $y^i \cdot x$ and $y \cdot x^i$ are of the same length (with $2(i+1)$ terms),
with the same leading term and of the same internal byte-size. 
Counting the byte-size of both expressions for $i = 1..10, 15, 20$ we obtain $8, 21, 40, 65, 96, 133, 176, 225, 280, 341, 736, 1281$.
Indeed this sequence coincides with octagonal numbers\footnote{\url{http://oeis.org/A000567}} shifted by 1. Hence the byte-size $s(i)$ of $y^i\cdot x$ is $3i^2+5i+2$.

Further on we look for some computation-specific patterns, for instance, during the
computation of a left Gr\"obner basis. We use the implementation of 
\textit{slim Gr\"obner basis} algorithm in \textsc{Singular} for highly resource-demanding and tasks like the computation of Bernstein-Sato polynomials with two different algorithms.
With the latter algorithms
one computes in the tensor product of model algebras (Weyl and shift algebras),
what suggests using formulas. By cashing below we mean the use of the multiplication table
for saving once computed elements.

We experiment with the following strategies: 1. Using formulas and caching 2. Using formulas without caching the results 3. Caching the results, obtained without formulas. The timings 
and maximal memory usage are collected in the Table 3. 
Since we are interested in the efficiency of the caching, we count the requests to
compute every needed elementary product $y^a \cdot x^b$ as above in a separate computation. 
In the process of computation of Bernstein-Sato polynomial of various Reiffen curves $f(x,y)\in\K[x,y]$, there is a complicated computation in the 
tensor product of two Weyl algebras and $\K[s]$. 
We count requests to compute $y^m x^n=x^n y^m$ in Table 4 
and the number of requests to compute $\partial_x^m x^n$ in the first Weyl algebra, where $[\partial_x,x]=1$ holds, in Table 5. 
All the data is available online from \verb?http://www.mathematik.uni-kl.de/?$\sim$\verb?motsak/ncSAtests?.
Timings and memory are given in seconds resp. in Kb. The tests were run on a PC running 64 bit Arch Linux 2.6.38, having 16 GB RAM and Intel Core i7 CPU 860 at 2.80GHz (4 Cores/8 Threads).

We gain some speedup by using both caching and formulas. It can be enhanced by
optimizing the way of caching, especially for algebras with many variables.

\begin{table}
\label{CompMult}
\caption{Time and Memory Comparison for Different Multiplication Strategies}
\begin{center}
\begin{tabular}{|l|c|c|c|c|c|c|c|}
\hline
Name & \multicolumn{2}{|c|}{Cache + Formulas} & \multicolumn{2}{|c|}{Formulas only} & \multicolumn{2}{|c|}{Cache only} \\
\cline{2-7}
 & time & memory & time & memory & time & memory \\
\hline 
reiffen11-mod& 8.04& 9.912 & 8.14 & 11.999 & 7.97 & 9.912 \\
reiffen45-3-ann& 8525.40& 518.262 & 8517.73& 523.502 & 8547.37& 518.261 \\
reiffen45-6-ann& 158.25& 46.509 & 160.20 & 46.506 & 160.54& 46.509 \\
reiffen57-mod& 36.16& 11.422 & 39.28 & 11.415 & 36.41 & 11.422 \\
reiffen59-mod& 80.03& 15.759 & 87.90 & 15.746 & 81.03 & 15.759 \\
reiffen67-mod& 216.81& 45.369 & 261.30 & 45.361 & 222.15& 45.369 \\
reiffen68-mod& 117.58& 142.254 & 141.97 & 138.494 & 118.31& 142.254 \\
reiffen76-mod& 389.54& 99.026 & 630.48& 99.026 & 389.42& 99.026 \\
reiffen86-mod& 298.83& 69.610 & 453.93& 69.610 & 297.81& 69.610  \\
\hline
\end{tabular}
\end{center}
\end{table}



\begin{table}
\label{CountXY}
\caption{Number of Requests for $y^m x^n=x^n y^m$}
\begin{center}
\begin{tabular}{|l|c|c|c|c|c|c|c|c|c|c|c|}
\hline
 & $n$ & 1 & 2 & 3 & 4 & 5 & 6 & 7 & 8 & 9 & 10 \\
\cline{2-12}
$m$ & & & & & & & & & & & \\

\hline 
1 & 23711 & 18629 & 17628 & 14796 & 8368 & 2899 & 2444 & 1315 & 296 & 186 & 32 \\
2 & 8264 & 4952 & 4806 & 4947 & 2952 & 728 & 715 & 549 & 47 & 26 & 0 \\
3 & 4900 & 3002 & 3233 & 3202 & 1577 & 286 & 277 & 237 & 18 & 9 & 0 \\
4 & 2084 & 1230 & 1268 & 1189 & 585 & 104 & 82 & 60 & 6 & 3 & 0  \\
5 & 215 & 155 & 118 & 131 & 127 & 59 & 48 & 37 & 2 & 1 & 0  \\
6 & 62 & 45 & 30 & 30 & 26 & 6 & 3 & 0 & 0 & 0 & 0  \\
7 & 19 & 14 & 9 & 9 & 8 & 2 & 1 & 0 & 0 & 0 & 0  \\
8 & 3 & 2 & 1 & 1 & 1 & 0 & 0 & 0 & 0 & 0 & 0 \\
\hline
\end{tabular}
\end{center}
\end{table}

\begin{table}
\label{CountXDX}
\caption{Number of Requests for $\partial_x^m x^n = x^n \partial_x^m + \ldots$}
\begin{center}
\begin{tabular}{|l|c|c|c|c|c|c|c|c|c|c|c|c|}
\hline
 & $n$ & 1 & 2 & 3 & 4 & 5 & 6 & 7 & 8 & 9 & 10 & 11 \\
\cline{2-13}
$m$ & & & & & & & & & & & & \\
\hline 
1 & 27345 & 22324 & 21914 & 20484 & 14636 & 5702 & 4076 & 3104 & 1515 & 1005 & 563 & 164 \\
2 & 12627 & 10267 & 9799 & 9219 & 6910 & 2592 & 1888 & 1523 & 718 & 455 & 246 & 72 \\
3 & 4271 & 3319 & 2904 & 2895 & 2544 & 942 & 763 & 691 & 300 & 181 & 90 & 26 \\
4 & 1149 & 872 & 604 & 659 & 780 & 277 & 273 & 275 & 105 & 58 & 24 & 6 \\
5 & 247 & 203 & 50 & 79 & 224 & 54 & 65 & 83 & 26 & 11 & 0 & 0 \\
\hline
\end{tabular}
\end{center}
\end{table}

\newpage
\section{Centers and Ring-Theoretic Properties of Model Algebras}

By using the formulas, we compute explicitly the centers of non-commutative model algebras, depending on the ground field $\K$. Recall, that for some $f\in A$ one defines the centralizer 
subalgebra $C(f) = \{ a\in A \mid fa=af \} \supseteq \K[f]$. 
\begin{proposition}
\label{Centerz}
For the algebras of Lie type one has
\begin{itemize}
\item[$\bullet$] If $\Char \K=0$, $Z(A_1) = \K$ and  $Z(S_1)=\K$.
\item[$\bullet$] If $\Char \K=p$, $Z(A_1) = \K[x^p,\d^p]$ and $Z(S_1)=\K[x^p-x,s^p]$.
\end{itemize}
For the quantum algebras one has
\begin{itemize}
\item[$\bullet$] If $q$ is not a root of unity, $Z(\K_q[x,y]) = \K(q)$ and $Z(A_1^{(q)})=\K(q)$.
\item[$\bullet$] If $q$ is a primitive root of unity of order $p$ over $\K$, 
$Z(\K_q[x,y]) = \K(q)[x^p,y^p]$ and $Z(A_1^{(q)})=\K(q)[x^p,\d^p]$.
\end{itemize}
\end{proposition}

\begin{proof}
Since all the proofs are similar, let us consider the $q$-Weyl algebra $A_1^{(q)}$. We compute the center as the intersection of two centralizers 
$Z(A) = C(x) \cap C(\d)$.
Since $A$ is a $\Z$-graded algebra (e.~g.~ with $\deg x=-1$, $\deg d = 1$), 
$C(x), C(\d)$ and $Z(A)$
are $\Z$-graded subalgebras. The $0$-th graded part of $A$ is 
$\K(q)[x\d]$. For $k\in\Z_+$, the $k$-th graded part of $A$ is $A_k=\K(q)[x\d] \d^k$ and 
$A_{-k}=\K(q)[x\d] x^k$.
By Thm. \ref{MainThm}, we see that
$\d^m x = q^m x \d^m + [m]_q \d^{m-1} \in A_{m-1}$ is homogeneous of degree $m-1$. 
Thus for $f = \sum_\alpha c_{\alpha}(x) \d^{\alpha}$
one has $0=fx-xf= \sum_\alpha c_{\alpha}(x) (\d^{\alpha} x - x \d^{\alpha})$.
Note, that $\d^{\alpha} x - x \d^{\alpha} = (q^{\alpha}-1) x \d^{\alpha} + [\alpha]_q \d^{\alpha-1}$ is graded. 
So, for all $\alpha$ $(q^{\alpha}-1) x \d^{\alpha} + [\alpha]_q \d^{\alpha-1} = 0$ 
, that is $q^{\alpha}=1$ for all $\alpha$. 
Hence $q^p=1$ implies
$C(x)=\K(q)[x,\d^p]$, $C(\d)=\K(q)[x^p,\d]$ 
and thus $Z(A) = \K(q)[x^p,\d^p]$.
\end{proof}


It is known, that over any field $A(q,\alpha, \beta, \gamma)$ is a $G$-algebra (or a PBW algebra) 
\cite{LS,BGV}. Thus it is a Noetherian domain of Gel'fand-Kirillov dimension 2, which is Cohen-Macaulay and Auslander-regular \cite{BGV}.
However, the global homological dimension is between 1 and 2.

\begin{proposition}
\label{gldim}
$\gldim A(q,\alpha, \beta, \gamma) = 1$ if and only if 
$\Char \K = 0$ and $A(q,\alpha, \beta, \gamma)$ is isomorphic to
the Weyl algebra.
\end{proposition}
\begin{proof}
Let $A = A(q,\alpha, \beta, \gamma)$. Because of Cohen-Macaulay property, $\gldim A = 2$ if and only if there exist a module $M$ of finite dimension over $\K(q)$. We look for $M=A/L$ for an ideal $L\subset A$.
Over $\K[x,y]$ and $\K_q[x,y]$ all 1-dimensional modules are described by ideals $\langle x-a, y-b \rangle$ for $a,b\in\K(q)$. In the shift algebra there are ideals $\langle x-a, s \rangle$ for $a\in\K$ while in the $q$-Weyl algebra these ideals are $\langle x-a, y- ((1-q)a)^{-1} \rangle$ for $a\in\K(q)^*$.\\
Consider the case when $A$ is the Weyl algebra.
If $\Char \K = p >0$, from Prop. \ref{Centerz} follows, that $I_p=\langle x^p,\d^p \rangle$ is a proper two-sided ideal and $A/I_p$ is finite dimensional, thus $\gldim A = 2$.
Now let $\Char \K = 0$. A module of Gel'fand-Kirillov dimension 0 is finite-dimensional over $\K$ and hence can be considered as a representation of an algebra. 
 Assume there exists some $m\in\N$ and two matrices $X,D \in Mat(m,\K)$ such that $x\mapsto X, \d \mapsto D$ is a representation, in other words a homomorphism of left $A_1$-modules. Thus $DX-XD=\id_m$ holds and also $0=tr(DX-XD)=m$, what is a contradiction, showing that there are no finite dimensional modules and the minimal Gel'fand-Kirillov dimension of a module over $A_1$ is thus 1 (which encompasses for instance holonomic left modules $A_1/\langle x \rangle$ and $A_1/\langle \d \rangle$).
\end{proof}

\begin{lemma}
\label{A1toKq}
For any field $\K$, there are no nonzero $\K$-algebra homomorphisms from $A_1(\K)$ to $\K_q[x,y]$ or to $\K[x,y]$.
\end{lemma}

 \begin{proof}
 Assume there is a homomorphism of $\K$-algebras $\phi:A_1(\K) \to \K_q[x,y]$. Thus there exists 
 $X = \phi(x), D=\phi(\d) \in \K_q[x,y]$, such that $DX-XD=1$. 
 Write $D=\sum_{\alpha} c_{\alpha} x^{\alpha_1} y^{\alpha_2}$ for $c_{\alpha} \in\K$ and $\N_0^2 \ni \alpha = (\alpha_1,\alpha_2)$. Analogously 
 $X=\sum_{\beta} d_{\beta} x^{\beta_1} y^{\beta_2}$.
Then in $\K_q[x,y]$ one has $DX-XD = \sum_{\alpha, \beta} c_{\alpha} d_{\beta}(q^{\beta_1 \alpha_2} - q^{\beta_2 \alpha_1}) x^{\alpha_1+\beta_1} y^{\alpha_2 + \beta_2}$ and the coefficient by $1=x^0 y^0$ vanishes. In the limit $q\to 1$, that is in $\K[x,y]$ we obtain $DX-XD=0$.
Hence the only homomorphism from $A_1(\K)$ to $\K[x,y]$ or to $\K_q[x,y]$ is 0.
 \end{proof}


\begin{proposition}
\label{Modelz}
Five model algebras are pairwise non-isomorphic over any field.
\end{proposition}
\begin{proof}
Let $\Char \K = 0$. From Prop. \ref{Centerz} we see that $A_1(\K),S_1(\K),A_1^{(q)}(\K) \not\cong \K[x,y]$. By Prop. \ref{gldim} and Lemma \ref{A1toKq} we conclude $S_1(\K),\K_q[x,y], A_1^{(q)}(\K) \not\cong A_1(\K)$.
For any field $\K$, $A_1^{(q)}(\K) \not\cong \K_q[x,y]$: 
let $U,V$ be affine subspaces of $\K^2$
of all 1-dimensional (thus irreducible) 
representations of both algebras in $\K$. Then
$U,V$ are zero sets of corresponding ideals $I=\langle (1-q)ab+1 \rangle$
and $J=\langle (1-q)cd \rangle = \langle c \rangle \cap \langle d \rangle$,
what implies $\K[U]\not\cong \K[V]$.
Since the variety $W\subset \K^2$ 
of 1-dimensional representations of $S_1$ is $W=\{(a,b) \mid ba=ab+b\} = \{(a,0) \mid a \in\K\}$ cannot be in bijection with either $U$ or $V$, 
$S_1$ is not isomorphic to $A_1^{(q)}(\K)$ or $\K_q[x,y]$.
 Also $\K[x,y]$ with $\K^2$ as the variety of 
1-dimensional representations
is not isomorphic to other model algebras for any $\K$.
%
Now, let $\Char \K = p$. Then $A_1(\K)$ has finite dimensional representations
since 
$m = \tr(1_{m \times m}) = \tr(DX-XD)=0$ for a $m \times m$ representation $X,D$ of $A_1(\K)$. Hence $p\mid m$ and the smallest irreducible representation is in dimension $p$. Thus $A_1(\K)$ cannot be isomorphic to other model algebras.
The remaining cases can be proved analogously.
\end{proof}
\noindent
\textbf{Future work} includes the study of Ore localizations of model algebras,  
for which no analog of the "five models" theorem is not known yet. 
Groups of linear endomorphisms and sets of linear antiendomorphisms of model
algebras are of interest as well. \\
By performing Ore localization on model algebras one gets interesting and ubiquitous algebras, for which the "five models" theorem is not known yet.\\
\noindent
\textbf{Acknowledgments}. The authors are grateful to 
Hans Sch\"onemann and Oleksandr Yena 
for discussions on the subject. We would like 
to thank anonymous referees for valuable suggestions.
The second author was supported by the Austrian Science Fund (FWF): P20162-N18.
The first and third authors are grateful to the SCIEnce project (Transnational access) at RISC for supporting their visits to RISC and the usage of computational infrastructure at RISC.

\begin{thebibliography}{99}
\bibitem{BGV}
J.~Bueso, J.~G{\'o}mez-Torrecillas, and A.~Verschoren.
\newblock {\em Algorithmic methods in non-commutative algebra. Applications to
  quantum groups}.
\newblock Kluwer Acad. Publ., 2003.

\bibitem{Mgfun}
F.~Chyzak and B.~Salvy.
\newblock Non--commutative elimination in {O}re algebras proves multivariate
  identities.
\newblock {\em J. Symbolic Computation}, 26(2):187--227, 1998.

\bibitem{Dix}
J.~Dixmier.
\newblock {\em Enveloping Algebras.}
\newblock AMS, 1996.

\bibitem{Gr}
E.~Green.
\newblock Multiplicative Bases, {G}r\"obner Bases, and Right {G}r\"obner Bases.
\newblock {\em J. Symbolic Computation}, 29(4-5):601--623, 2000.

\bibitem{HolFun}
C.~Koutschan.
\newblock {HolonomicFunctions (User's Guide)}.
\newblock Technical Report 10-01, RISC Report Series, University of Linz,
  Austria, 2010.

\bibitem{LS}
V.~Levandovskyy and H.~Sch{\"o}nemann.
\newblock Plural --- a computer algebra system for noncommutative polynomial
  algebras.
\newblock In {\em Proc. ISSAC}, 176--183. ACM Press, 2003.

\bibitem{MR}
J.~McConnell and J.~Robson.
\newblock {\em Noncommutative Noetherian rings.}
\newblock AMS, 2001.

\bibitem{Mora}
T.~Mora
\newblock An introduction to commutative and non-commutative {G}r\"obner bases.
\newblock {\em Theor. Comp. Sci.}, 134:131--173, 1994.

\bibitem{Ufn98}
V.~Ufnarovski
\newblock Introduction to noncommutative {G}r\"obner bases theory.
\newblock In B.~Buchberger and F.~Winkler (eds.) 
  {\em {G}r\"obner bases and applications}, 259--280, 1998.
\end{thebibliography}

\end{document}